\documentclass[journal,draftcls,onecolumn,11pt,twoside]{IEEEtran}
\normalsize
\usepackage{amsmath,amssymb,amsthm}
\usepackage{enumerate}
\usepackage{stfloats}
\usepackage{comment}
\usepackage{subcaption}
\usepackage{siunitx}
\usepackage{cite}

\usepackage[dvipsnames]{xcolor}
\definecolor{myPink}{RGB}{255,105,183}

\usepackage[T1]{fontenc}
\usepackage{graphics} 
\usepackage{epsfig} 
\usepackage[mathscr]{euscript}
\usepackage{algorithm}
\usepackage[noend]{algpseudocode}
\makeatletter
\def\BState{\State\hskip-\ALG@thistlm}
\makeatother

\usepackage{soul}
\usepackage{tikz}
\usetikzlibrary{arrows,shapes,chains,matrix,positioning,scopes,patterns,calc,
decorations.markings,
decorations.pathmorphing,
}

\usepackage{pgfplots}
\pgfplotsset{compat=1.3}
\usepgflibrary{shapes}

\newtheorem{theorem}{Theorem}
\newtheorem{lemma}[theorem]{Lemma}

\newtheorem{definition}[theorem]{Definition}

\renewcommand{\epsilon}{\varepsilon}

\newcommand{\RNum}[1]{\uppercase\expandafter{\romannumeral #1\relax}}

\newcommand{\dv}{\vec{\mathrm{d}}}

\DeclareMathAlphabet{\mcl}{OMS}{cmsy}{m}{n}

\newlength\tikzwidth
\newlength\tikzheight

\definecolor{mycolor1}{rgb}{0.63529,0.07843,0.18431}%
\definecolor{mycolor2}{rgb}{0.00000,0.44706,0.74118}%
\definecolor{mycolor3}{rgb}{0.00000,0.49804,0.00000}%
\definecolor{mycolor4}{rgb}{0.87059,0.49020,0.00000}%
\definecolor{mycolor5}{rgb}{0.00000,0.44700,0.74100}%
\definecolor{mycolor6}{rgb}{0.74902,0.00000,0.74902}%

\newif\iflonger
\longertrue

\newif\ifhashcoll
\hashcollfalse

\newif\ifproof
\prooftrue
\usepackage{pgfplots}
\usepackage{tikz}
\usepackage{cite}
\usepackage{amsmath}
\usepackage{amssymb}
\def\fig_path{./Figures}
\begin{document}
\title {{Collaborative Decoding of Polynomial Codes for Distributed Computation}}

%
\author{\begin{tabular}{c} Adarsh M. Subramaniam,
Anoosheh Heiderzadeh,
Krishna R. Narayanan
\end{tabular} \\
Department of Electrical and Computer Engineering, \\
Texas A\&M University
}
\maketitle

\begin{abstract}
We show that polynomial codes (and some related codes) used for distributed matrix multiplication are {\em interleaved} Reed-Solomon codes and, hence, can be collaboratively decoded. We consider a fault tolerant setup where $t$ worker nodes return erroneous values. For an additive random Gaussian error model, we show that for all $t < N-K$, errors can be corrected with probability 1. Further, numerical results show that in the presence of additive errors, when $L$ Reed-Solomon codes are collaboratively decoded, the numerical stability in recovering the error locator polynomial improves with increasing $L$.
\end{abstract}


\section{Introduction and Main Result}
We consider the problem of computing $\mathbf{A}^{\mathsf{T}}\mathbf{B}$ for two matrices $\mathbf{A} \in \mathbb{F}^{s \times r}$ and $\mathbf{B} \in \mathbb{F}^{s \times r'}$ (for an arbitrary field $\mathbb{F}$)\footnote{Some results in this paper will apply to specific fields and this will be clarified later.} in a distributed fashion with $N$ worker nodes using a coded matrix multiplication scheme \cite{yu2018polycode,yu2018straggler,dutta2018unified,unidecode,lagrange_coded,speeding_up,shortdot,coded_unified,highdim,stserv_mult,inv_pro}
To keep the presentation clear, we will focus on one class of codes, namely Polynomial codes, and explain our results in relation to the Polynomial codes~\cite{yu2018polycode}; notwithstanding, our results also apply to Entangled Polynomial codes \cite{yu2018straggler} and PolyDot codes \cite{dutta2018unified}. 
We assume that the matrices $\mathbf{A}$ and $\mathbf{B}$ are split into $m$ subblocks and $n$ subblocks, respectively. These subblocks are encoded using a Polynomial code~\cite{yu2018straggler}. Each worker node performs a matrix multiplication and returns a matrix with a total of $L=\frac{rr'}{mn}$ elements (from $\mathbb{F}$) to the master node. 

Our main interest is in the fault-tolerant setup where some of the $N$ worker nodes return erroneous values. We say that an error pattern of Hamming weight $t$ has occurred if $t$ worker nodes return matrices that contain some erroneous values. 
The main idea in the Polynomial codes, Entangled Polynomial codes and PolyDot codes is to encode the subblocks of $\mathbf{A}$ and $\mathbf{B}$ in a clever way such that the matrix product returned by the worker nodes are symbols of a codeword of a Reed-Solomon (RS) code over $\mathbb{F}$. The properties of an RS code are then used to obtain bounds on the error-correction capability of the scheme.

The main contribution of this work relies on the observation that Polynomial codes, Entangled Polynomial codes, and PolyDot codes are not just RS codes, but an Interleaved Reed-Solomon (IRS) code which consists of several RS codes that can be collaboratively decoded (see Section~\ref{poly ir} or~\cite{schmidt2009collaborative} for a formal definition). 
This additional structure provides the opportunity for collaborative decoding of multiple RS codes involved in such coded matrix multiplication schemes. Such a collaborative decoding, for which efficient multi-sequence shift-register (MSSR) based decoding algorithms exist~\cite{krachkovsky1997decoding}, provides a practical decoder with quadratic complexity in $t$, while potentially nearly doubling the decoding radius.

The main results of this paper and their relation to the existing results are as follows. In \cite{yu2018straggler}, it is shown that any error pattern with Hamming weight $t$ can be corrected if $t \leq  \lfloor\frac{N-K}{2}\rfloor$
where $K = mn$ is the effective dimension of the Polynomial code. Very recently, Dutta {\em et al.} in~\cite{dutta2018unified} showed that when $\mathbb{F} = \mathbb{R}$ (the real field) and error values are randomly distributed according to a Gaussian distribution, with probability $1$ all error patterns of Hamming weight $t \leq N-K-1$ can be corrected.  To attain this bound,~\cite{dutta2018unified} uses a decoding algorithm which is similar in spirit to exhaustive maximum likelihood decoding with a complexity that is $O\big(LN^{\min\{t,N-t\}}\big)$. This can be prohibitive for many practical values of $N$ and $t$. 
In \cite{dutta2018unified}, it is suggested that in practice, the performance of ML decoding can be approximated by algorithms with 
polynomial complexity in $N$ such as the $\ell_1$-minimization algorithm~\cite{candes2005decoding}.
However, there is no proof (nor evidence) that such algorithms can correct all error patterns of Hamming weight up to $N-K-1$ with probability $1$. Indeed, as we will show in this work, the standard $\ell_1$-minimization based decoding algorithm~\cite{candes2005decoding} fails to correct all error patterns of Hamming weight up to $N-K-1$ with a non-zero probability.

In this work, we show that we can use the MSSR decoding algorithm of~\cite{krachkovsky1997decoding} for decoding Polynomial codes with the complexity of $O\big(Lt^2+N\big)$. For this algorithm, we will show that when $\mathbb{F} = \mathbb{F}_q$ (a finite field with $q$ elements),
for $\lfloor\frac{N-K}{2}\rfloor  < t \leq \frac{L}{L+1} (N-K)$, all but a fraction $\gamma(t)$ of the error patterns of Hamming weight $t$ can be corrected where $\gamma(t) \rightarrow 0$ as $q \rightarrow \infty$. 
In particular, the convergence of $\gamma(t)$ to zero is exponentially fast in $L$, i.e., $\gamma(t) = q^{-\Omega(L)},$ for $\lfloor \frac{N-K}{2} \rfloor < t \leq \frac{L}{L+1} (N-K)$. In addition, when $\mathbb{F} = \mathbb{R}$, by extending the results of~\cite{krachkovsky1997decoding} and~\cite{schmidt2006} to the real field and using the results of \cite{dutta2018unified}, we will show that for $L \geq N-K-1$ and ${\lfloor\frac{N-K}{2}\rfloor< t \leq N-K-1}$, all error patterns of Hamming weight $t$ can be corrected with probability $1$, under the random Gaussian error model previously considered in \cite{dutta2018unified}.


In a nutshell, our results show that with a probability arbitrarily close to $1$ (or respectively, with probability $1$), all error patterns of Hamming weight up to $\frac{L}{L+1}(N-K)$, which can be made arbitrarily close to $N-K-1$ for sufficiently large $L$, can be corrected for sufficiently large finite fields (or respectively, the real field). Not only does this indicate a substantial increase in the error-correction radius with provable guarantees when compared to the results in \cite{yu2018straggler}, 
but it also shows that the Dutta {\em et al}.'s upper bound in \cite{dutta2018unified} can be achieved with a practical decoder with a quadratic complexity in the number of faulty worker nodes ($t$). This improvement in complexity is the result of collaboratively decoding the IRS code instead of separately decoding the RS codes using a maximum likelihood decoder as is done in~\cite{dutta2018unified}.


\section{Review of Polynomial Codes for Distributed Matrix Multiplication}
\subsection{Notation} 
Throughout the paper, we denote matrices by boldface capital letters, e.g., $\mathbf{A}$, and denote vectors by boldface small letters, e.g., $\mathbf{a}$. For an integer $i\geq 1$, we denote $\{1,\dots,i\}$ by $[i]$, and for two integers $i$ and $j$ such that $i<j$, we denote $\{i,i+1,\dots,j\}$ by $[i,j]$.  We use the short notation $((f(i,j))_{i\in [m],j\in [n]})$ to represent an $m\times n$ matrix whose entry $(i,j)$ is $f(i,j)$, where $f(i,j)$ is a function of $i$ and $j$. We occasionally use the compact notation $(\mathbf{a}_1,\mathbf{a}_2,\dots,\mathbf{a}_{n})$ to represent an $m\times n$ matrix whose columns are the column-vectors $\mathbf{a}_1,\mathbf{a}_2,\dots,\mathbf{a}_n$, each of length $m$. Similarly, sometimes we use the compact notation $(\mathbf{a}_1;\mathbf{a}_2;\dots;\mathbf{a}_m)$ to represent an $m\times n$ matrix whose rows are the row-vectors $\mathbf{a}_1,\mathbf{a}_2,\dots,\mathbf{a}_m$, each of length $n$. We also denote by $\mathbf{A}(i,:)$ and $\mathbf{A}(:,j)$ the $i$th row and the $j$th column of a matrix $\mathbf{A}$, respectively. A vector or a matrix with a $\wedge$ above is an estimate.


\subsection{Polynomial Codes}
In this section, we review the Polynomial codes of Yu, Maddah-Ali and Avestimehr \cite{yu2018polycode} for distributed matrix multiplication. Consider the problem of computing $\mathbf{A}^{\mathsf{T}} \mathbf{B}$ in a distributed fashion for two matrices $\mathbf{A} \in \mathbb{F}^{s \times r}$ and $\mathbf{B} \in \mathbb{F}^{s \times r'}$ for an arbitrary field $\mathbb{F}$. In the scheme of Polynomial codes in~\cite{yu2018polycode}, the master node distributes the task of matrix multiplication among $N$ worker nodes as follows. 

The columns of $\mathbf{A}$ and $\mathbf{B}$ are first partitioned into $m$ partitions $\mathbf{A}_0,\mathbf{A}_1,\dots,\mathbf{A}_{m-1}$ of equal size $\frac{r}{m}$ and $n$ partitions $\mathbf{B}_0,\mathbf{B}_1,\dots,\mathbf{B}_{n-1}$ of equal size $\frac{r'}{n}$, respectively,
$$\mathbf{A}=[\mathbf{A}_0\  \mathbf{A}_1 \cdots \mathbf{A}_{m-1}], \ \ \mathbf{B}=[\mathbf{B}_0 \ \mathbf{B}_1 \cdots \mathbf{B}_{n-1}].$$

Let $x_1,x_2,\ldots,x_N$ be $N$ distinct elements in $\mathbb{F}$. 
For two parameters $\alpha,\beta\in [N]$, let $\tilde{\mathbf{A}}_i$ and $\tilde{\mathbf{B}}_i$ be matrices defined by,
\[\tilde{\mathbf{A}}_i=\sum\limits_{j=0}^{m-1} \mathbf{A}_j x_{i}^{j\alpha}, \ \ 
\tilde{\mathbf{B}}_{i}=\sum\limits_{j=0}^{n-1} \mathbf{B}_j x_i^{j\beta}.\]
The dimensions of the matrices $\tilde{\mathbf{A}}_i$ and $\tilde{\mathbf{B}}_i$ are $s \times \frac{r}{m}$ and $s \times \frac{r'}{n}$, respectively. 

The $i$th worker node computes the smaller matrix product $\tilde{\mathbf{C}}_i$ given the values of $\tilde{\mathbf{A}}_i$ and $\tilde{\mathbf{B}}_i$, 
\begin{equation}\label{eq prod}
\tilde{\mathbf{C}}_i=\tilde{\mathbf{A}}_i^{\mathsf{T}}\tilde{\mathbf{B}}_i = \sum\limits_{j=0}^{m-1} \sum\limits_{k=0}^{n-1} \mathbf{A}_j^{\mathsf{T}} \mathbf{B}_k \  {x_i^{j\alpha + k\beta}}. 
\end{equation}
The parameters $\alpha$ and $\beta$ are chosen carefully such that for each pair $(j,k)$ the corresponding exponent of $x_i$ (i.e., $j\alpha+k\beta$) is distinct. For instance, one such choice for $\alpha$ and $\beta$ is $\alpha=1$ and $\beta=m$. In this case, the $i$th worker node essentially evaluates $\mathbf{P}(x)$ at $x=x_i$ and returns $\mathbf{P}(x_i)$, where 
\begin{equation} \label{eq poly}
\mathbf{P}(x)=\sum\limits_{j=0}^{m-1} \sum\limits_{k=0}^{n-1} \mathbf{A}_j^{\mathsf{T}} \mathbf{B}_k \ x^{j+km}.    
\end{equation}
The coefficients in the polynomial $\mathbf{P}(x)$ are the $mn$ uncoded symbols of the product $\tilde{\mathbf{C}}_i$ in (\ref{eq prod}). The crux of the Polynomial code is that the vector of coded symbols $\left(\mathbf{P}(x_1),\ldots,\mathbf{P}(x_N)\right)=(\tilde{\mathbf{C}}_1,\tilde{\mathbf{C}}_2,\cdots,\tilde{\mathbf{C}}_N)$ can be considered as a codeword of a Reed-Solomon (RS) code. If $N$ worker nodes are available in the distributed system, a Polynomial code essentially evaluates the polynomial $\mathbf{P}(x)$ at $N$ points of the field $\mathbb{F}$; any $mn$ of which can recover the coefficients which can be put together to recover the matrix product. The minimum number of worker nodes that need to compute and return the correct evaluations of $\mathbf{P}(x)$ for the master node to be able to successfully recover the matrix product $\mathbf{A}^{\mathsf{T}} \mathbf{B}$ is called the \emph{recovery threshold}. Viewing the recovery process of a Polynomial code as a polynomial interpolation operation, it can be seen that the recovery threshold of the Polynomial code is $mn$~\cite{yu2018polycode}.

\section{Polynomial Codes are Interleaved Reed-Solomon Codes} \label{poly ir}

\begin{definition}{Generalized Reed-Solomon (GRS) Codes:} Let $\mathbf{m} = (m_0,m_1,\ldots,m_{K-1})$ and let the associated polynomial $m(x)$ be defined as $m(x) := m_0+m_1 x+\ldots+m_{K-1}x^{K-1}$. Further, let
$\mathbf{c} = (c_0,c_1,\dots,c_{N-1})$,
$\pmb{\alpha} = (\alpha_0,\alpha_1,\ldots,\alpha_{N-1})$ and
$\mathbf{v} = (v_0,v_1,\ldots,v_{N-1})$ be three row vectors such that $c_i, \alpha_i, v_i \in \mathbb{F}$, $v_i \neq 0$, and $\alpha_i \neq \alpha_j$.
A Generalized Reed-Solomon (GRS) code $\mathcal{C}$ over $\mathbb{F}$ of length $N$, dimension $K$, evaluation points $\pmb{\alpha}$, weight vectors $\mathbf{v}$, denoted by \emph{GRS$(\mathbb{F},N,K,\pmb{\alpha},\mathbf{v})$}, is the set of all row-vectors (codewords) $\mathbf{c} = (v_0 m(\alpha_0), v_1 m(\alpha_1), \ldots, v_{N-1} m(\alpha_{N-1}))$, i.e., $c_i =v_i  m(\alpha_i)$. Equivalently, a GRS code is also the set of codewords $\mathbf{c}$ such that for all $i\in [0,N-K-1]$, $\sum_{j=0}^{N-1} u_j c_j (\alpha_i)^j = 0$, where $u_i^{-1} = v_i \prod\limits_{j \neq i} (\alpha_i-\alpha_j)$.
The minimum distance of such a GRS code is $d_{\min} = N-K+1$.
\end{definition}

Reed-Solomon (RS) codes are a special case of GRS codes with $v_i = 1, u_i = 1, \forall i \in [0,N-1]$. For finite fields and the complex field, an $\pmb{\alpha}$ exists such that $v_i = 1$ and $u_i = 1$, $ i \in [0,N-1]$. However for the real field, $u_i$ and $v_i$ cannot be simultaneously set to 1 and, hence, it is required to consider GRS codes. 

\begin{definition}{Interleaved Generalized Reed-Solomon (IGRS) Codes} \cite{schmidt2009collaborative}: Let $\{\mathcal{C}^{(l)}\}_{l\in [L]}$ be a collection of $L$ GRS codes $\mathcal{C}^{(l)}\triangleq \emph{\text{RS}}(\mathbb{F},N,K^{(l)},\pmb{\alpha},\pmb{u})$, each of length $N$ over a field $\mathbb{F}$, where the dimension and minimum distance of the $l$th GRS code are $K^{(l)}$ and $d^{(l)}$, respectively. Then, an Interleaved Generalized Reed-Solomon (IGRS) code $\mathcal{C}_{\rm{IGRS}}$ is the set of all $L\times N$ matrices $(\mathbf{c}^{(1)};\mathbf{c}^{(2)};\dots;\mathbf{c}^{(L)})$ 
where $\mathbf{c}^{(l)} \in \mathcal{C}^{(l)}$ for $l\in [L]$ \cite{krachkovsky1997decoding}. If all the $L$ GRS codes $\mathcal{C}^{(l)}$ are equivalent, i.e., $\mathcal{C}^{(l)}=\mathcal{C}$ for all $l\in [L]$, the IGRS code $\mathcal{C}_{\mathrm{IRS}}$ is called \emph{homogeneous}.
\end{definition}

The chief observation in this work is that the Polynomial codes, Entangled Polynomial codes, and PolyDot codes are IGRS codes. Here, we formally prove this observation for the Polynomial codes. We shall henceforth refer to GRS codes and IGRS codes as RS codes and IRS codes, respectively. 

\begin{theorem}
A Polynomial code is an IRS code.
\end{theorem}

\begin{proof}
Let $\mathbf{W}$ be an $a\times b$ matrix with entries from $\mathbb{F}$, and let $\Gamma:\mathbb{F}^{a\times b} \rightarrow \mathbb{F}^{ab}$ denote a vectorizing operator which reshapes a matrix $\mathbf{W}$ into a column-vector $\mathbf{w}=(w_1,\dots,w_{ab})^{\mathsf{T}}$, i.e., $\Gamma(\mathbf{W})=\mathbf{w}$, such that $w_{(i-1)b+j}=\mathbf{W}(i,j)$, where $\mathbf{W}(i,j)$ is the element $(i,j)$ of $\mathbf{W}$. 

Let $\tilde{\mathbf{C}}_i(p,q)$ be the element $(p,q)$ of the matrix $\tilde{\mathbf{C}}_i$, 
\begin{equation} \label{eq poly irs}
\tilde{\mathbf{C}}_i(p,q)\triangleq\sum\limits_{j=0}^{m-1} \sum\limits_{k=0}^{n-1} [\mathbf{A}_j^{\mathsf{T}} \mathbf{B}_k]{(p,q)} x_i^{j+km}.  
\end{equation} Consider the $\frac{rr'}{mn} \times N$ matrix $\mathbf{D}\triangleq (\Gamma(\tilde{\mathbf{C}}_1),\Gamma(\tilde{\mathbf{C}}_2),\dots,\Gamma(\tilde{\mathbf{C}}_{N}))$, 
where the $i$th column of $\mathbf{D}$, namely $\Gamma(\tilde{\mathbf{C}}_i)$, is obtained by applying the vectorizing operator $\Gamma$ to $\tilde{\mathbf{C}}_i$. Let $(p_i,q_i)$ be the unique pair $(p,q)$ such that $i=(p-1)\frac{r'}{n}+q$. Then, the element $(i,j)$ of $\mathbf{D}$ is $\tilde{\mathbf{C}}_j{(p_i,q_i)}$, and accordingly, 
the $i$th row of $\mathbf{D}$ is given by $[\tilde{\mathbf{C}}_1{(p_i,q_i)}, \tilde{\mathbf{C}}_2{(p_i,q_i)}, 
\ldots,\tilde{\mathbf{C}}_N{(p_i,q_i)}]$, which is a codeword of an RS code. Thus the matrix $\mathbf{D}$ is a codeword of an IRS code with $L = \frac{rr'}{mn}$. In particular, the $i$th worker node computes $\tilde{\mathbf{C}}_i$ that has dimension $\frac{r}{m}\times\frac{r'}{n}$. 
It is evident from (\ref{eq poly irs}) that the element $(p,q)$ of $\tilde{\mathbf{C}}_i$ is the message polynomial $ \sum\limits_{j=0}^{m-1} \sum\limits_{k=0}^{n-1} [\mathbf{A}_j^{\mathsf{T}} \mathbf{B}_k]{(p,q)} x^{j+km} $ evaluated at $x_i$. Thus, $\tilde{\mathbf{C}}_i$ contains $\frac{rr'}{mn}$ RS codes evaluated at $x_i$ by the $i$th worker node. That is, the computations returned by the $i$th worker node constitute the $i$th column of an IRS code with $N$ being the number of worker nodes and $L=\frac{rr'}{mn}$ being the number of RS codes. This shows that a Polynomial code is a homogeneous IRS code with $K^{(l)}=mn$ for $l\in [L]$.
\end{proof}



\subsection{Error Matrix and Error Models} 
\label{fte} We consider the case when the worker nodes introduce additive errors in their computation. Let ${\mathbf{E}_i\in \mathbb{F}^{\frac{r}{m}\times\frac{r'}{n}}}$ denote the error matrix introduced by the $i$th worker node. Then the master node receives the set of matrices $\tilde{\mathbf{R}}_i$, for $i\in [N]$ where
$\tilde{\mathbf{R}}_i = \tilde{\mathbf{C}}_i \oplus \tilde{\mathbf{E}}_i$. Let $\mathbf{R}$ be the $\frac{rr'}{mn}\times N$ matrix of values received by the master node where the $i$th column of $\mathbf{R}$ is given by $\Gamma(\tilde{\mathbf{R}}_i)$, and let $\mathbf{E}$, referred to as the \emph{error matrix}, be the $\frac{rr'}{mn}\times N$ matrix of error values where the $i$th column of $\mathbf{E}$ is given by $\Gamma(\tilde{\mathbf{E}}_i)$. Then, $\mathbf{R}=\mathbf{D}\oplus\mathbf{E}$
where $\mathbf{D}$ is a codeword of an IRS code. If the $i$th worker node returns erroneous values, then the $i$th column of $\mathbf{R}$ will contain errors. Thus, 
the original problem of fault-tolerant distributed matrix multiplication reduces to the problem of decoding $\mathbf{D}$ from $\mathbf{R}$.

\begin{definition}
The Hamming weight of the matrix $\mathbf{E}$ denoted by $W_H(\mathbf{E})$ is defined as the number of non-zero columns in $\mathbf{E}$.
\end{definition}

We consider two different error models. First, we consider the Uniform Random Error for Finite Fields (UREF) model where the non-zero columns of the error matrix $\mathbf{E}$ are assumed to be uniformly distributed over all the non-zero vectors in $\mathbb{F}^L_{q}$ for a finite field $\mathbb{F}_q$. We further extend this model to the real field $\mathbb{R}$ where each non-zero entry in the error matrix $\mathbf{E}$ is assumed to be an independently and identically distributed Gaussian random variable (with arbitrary mean and variance). 
This model is referred to as the Gaussian Random Error (GRE) model. 

\subsection{Decoding and Error Events}

Let $\psi: \mathbb{F}^{L \times N} \rightarrow \{\mathcal{C}_{\rm{IRS}},F\}$ be the decoding function, where $F$ is a symbol that denotes decoding failure. 
A \emph{decoding error} is said to have occurred if $\psi(\mathbf{R}) \neq \mathbf{D}$. An \emph{undetected decoding error} is said to have occurred if $\psi(\mathbf{R}) \neq \mathbf{D}$ and $\psi(\mathbf{R})\neq F$, whereas a \emph{decoding failure} is said to have occurred if $\psi(\mathbf{R})=F$.

\section{Collaborative Decoding of Interleaved Reed-Solomon Codes}
\indent Simultaneous decoding of all the RS codes in an IRS code is known as \emph{collaborative decoding}. As shown in~\cite{schmidt2009collaborative} and~\cite{krachkovsky1997decoding}, collaborative decoding of IRS codes has certain advantages. In particular, when burst errors occur, they occur on the same column of the IRS code. Hence, multiple RS codewords share the same error positions. 
Note that an IRS code is actually a set of RS codes stacked together, each of which yields a set of syndrome equations. Intuitively, when burst errors occur, the error locator polynomials are more or less the same for all the RS codes but the number of syndrome equations increases with the number of stacked RS codes. This implies that a much larger set of errors can be corrected. This is because the rank of the stacked syndrome matrix is greater than or equal to the rank of the individual syndrome matrices, thus giving rise to the possibility of a greater decoding radius than the unique decoding bound of $\frac{1-R}{2}$, where $R$ is the code rate. More specifically, it was shown by Schmidt {\em et al.} in~\cite{schmidt2009collaborative} that when a set of $L$ RS codes are collaboratively decoded, except for a small probability of failure and a small probability of error (discussed in Section~\ref{pe_pf_section}), the fraction of errors that can be corrected can be as large as $\frac{L}{L+1}(1-R)$. 

\section{Decoding Algorithms} \label{sec_alg}
\subsection{Collaborative Peterson's Algorithm}
In this section, we propose a collaborative version of the Peterson's algorithm \cite{moon2005error} to
correct up to $t \leq t_{\text{max}}\triangleq  \frac{L}{L+1}(N-K)$ errors.

Consider $t$ non-zero errors in columns $j_1,j_2,\dots,j_t$ of the matrix $\mathbf{R}$ (i.e., the indices of the non-zero columns of the error matrix $\mathbf{E}$ are $j_1,j_2,\dots,j_t$).
Let $r^{(l)}(z) \triangleq \sum_{j=0}^{N-1} u_j \mathbf{R}(l,j) z^{j-1}$ be the modified (multiplying component-wise by $u_j$) received polynomial for the $l$th RS code, where $\mathbf{R}(l,j)$ is the element $(l,j)$ of the matrix $\mathbf{R}$. Then, the syndrome sequence for the $l$th RS code is given by $S^{(l)} \triangleq \{S^{(l)}_i\}_{i=0}^{N-K-1}$, where $S_i^{(l)} \triangleq \sum_{j=0}^{N-1} u_j \mathbf{R}(l,j) \alpha_j^i$ for $i\in[0,N-K-1]$. Define the error locator polynomial $\Lambda(z)$ as
\[
\Lambda(z) \triangleq \prod\limits_{i=1}^t (1- z \alpha_{j_i} ) = 1+\lambda_1 z + \dots + \lambda_t z^t
\]
and let $\boldsymbol{\lambda}(t) = (\lambda_t, \lambda_{t-1}, \dots, \lambda_1)^{\mathsf{T}}$ be the error locator vector associated with the error locator polynomial $\Lambda(z)$. When $t$ errors occur $\Lambda(z)$ has a degree of $t$.
The syndrome matrix $\mathbf{S}^{(l)}(t)$ and a vector $\mathbf{a}^{(l)}(t)$ for the $l$th RS code are given by 
\begin{equation}\label{eq:SlTl}\mathbf{S}^{(l)}(t)\triangleq\begin{pmatrix}
S^{(l)}_0 & S^{(l)}_1 & \cdots & S^{(l)}_{t-1} \\
S^{(l)}_1 & S^{(l)}_2 & \cdots & S^{(l)}_{t} \\
\vdots & \vdots &  & \vdots \\
S^{(l)}_{N-K-t-1} & S^{(l)}_{N-K-t} & \cdots & S^{(l)}_{N-K-2} \\
\end{pmatrix}, \ 
\mathbf{a}^{(l)}(t) \triangleq \begin{pmatrix}
-S^{(l)}_t \\
-S^{(l)}_{t+1} \\
\vdots \\
-S^{(l)}_{N-K-1}
\end{pmatrix}
\end{equation} Now we can write the following consistent linear system of equations for the IRS code,
\begin{equation}\label{eq:4}
\underbrace{
\begin{pmatrix}
\mathbf{S}^{(1)}(t)\\
\mathbf{S}^{(2)}(t)\\
\vdots \\
\mathbf{S}^{(L)}(t)
\end{pmatrix}}_{\mathbf{S}_L(t)}  \underbrace{
\begin{pmatrix}
\lambda_t\\
\lambda_{t-1} \\
\vdots \\
\lambda_{1}
\end{pmatrix}}_{\boldsymbol{\lambda}(t)} = \underbrace{
\begin{pmatrix}
\mathbf{a}^{(1)}(t)\\    
\mathbf{a}^{(2)}(t)\\  
\vdots \\
\mathbf{a}^{(L)}(t)
\end{pmatrix}
}_{\mathbf{a}_{L}(t)}
\end{equation} where $\mathbf{S}_L(t)$, the syndrome matrix for the IRS code, is the stacked matrix of $\mathbf{S}^{(l)}(t)$ for $l\in [L]$, and $\mathbf{a}_L(t)$, a vector for the IRS code, is the stacked vector of $\mathbf{a}^{(l)}(t)$ for $l\in [L]$. 
If $t$ columns of the matrix $\mathbf{R}$ are in error, then the error locator vector $\boldsymbol{\lambda}(t)$ can be obtained by the collaborative Peterson's algorithm, described in Algorithm~\ref{pet_alg}. The complexity of computing the rank of $\mathrm{rank}
(\mathbf{S}_L(\tau))$ is $O(L\tau^3)$; computing $\hat{\boldsymbol{\lambda}}$ requires $O(\tau^3)$ operations if the structure of $\mathbf{S}_L(\tau)$ is not exploited, and the Chien search has a complexity of $O(N)$. Since we have to consider all values of $\tau \in [t_{\max}]$, the overall complexity is $O(Lt_{\max}^4+N)$.

\begin{definition}
($t$-valid polynomial $\Lambda(z)$): A polynomial $\Lambda(z)$ over $\mathbb{F}$ is called $t$-valid if it is a polynomial of degree $t$ and possesses exactly $t$ distinct roots in $\mathbb{F}$. 
\end{definition}

\begin{algorithm}
\caption{Collaborative Peterson's algorithm for IRS Decoding}
\textbf{Input}: $S^{(l)}=\{S^{(l)}_i\}_{i=0}^{N-K-1} \ \forall  l\in [L]$\\
\textbf{Output}:  $\hat{\mathbf{D}} \in \{\mathbb{F}^{L \times N}, F \text{ (decoding failure)}\}$

\begin{algorithmic}[1]
\State $\hat{\mathbf{D}} = F$
\If {$\mathbf{S}_L(t) = \mathbf{0}$}  

\State $\hat{\mathbf{D}} = \mathbf{R}$ 

\Else
\For{each $t$ from 1 to $t_{\rm max}$}
\If {$\mathrm{rank}(\mathbf{S}_{L}^{\mathsf{T}}(t) \mathbf{S}_L(t)) = t$}
\State {$\hat{\boldsymbol{\lambda}} = (\mathbf{S}_{L}^{\mathsf{T}}(t) \mathbf{S}_L(t))^{-1} \mathbf{S}_{L}^{\mathsf{T}}(t)) \mathbf{a}_L(t)$}
\If {$\mathbf{S}_{L}(t)\ \hat{\boldsymbol{\lambda}} = \mathbf{a}_{L}(t)$}
\State {$(\hat{\lambda}_t, \hat{\lambda}_{t-1}, \dots, \hat{\lambda}_1) = \hat{\boldsymbol{\lambda}}^{\mathsf{T}}$}
\State {$\hat{\Lambda}(z) = 1+\hat{\lambda}_1z+\dots+\hat{\lambda}_t z^t$}
\If{$\hat{\Lambda}(z)$ is $t$-valid}
\State {Compute error locations $\hat{j}_i, \hat{j}_2, \ldots, \hat{j}_t$ using a Chien search \cite{moon2005error}}
\For{each $l$ from $1$ to $L$}
\State {From $\hat{j}_1, \ldots, \hat{j}_t$, and $S^{(l)}$, compute $\hat{\mathbf{E}}{(l,:)}$ using Forney's~algorithm~\cite{moon2005error}}
\State {Compute $\hat{\mathbf{D}}(l,:) = {\mathbf{R}}(l,:) - \hat{\mathbf{E}}{(l,:)}$}
\EndFor
\EndIf
\EndIf
\EndIf
\EndFor
\EndIf
\end{algorithmic}
\label{pet_alg}
\end{algorithm}

\subsection{Multiple Sequence Shift Register algorithm}
A more computationally efficient decoding algorithm 
to achieve error correction up to $t \leq t_{\mathrm{max}}=\frac{L}{L+1}(N-K)$ is the Multiple Sequence Shift Register (MSSR) algorithm proposed by Schmidt {\em et al.} in \cite{schmidt2006}. This algorithm has a complexity of $O(Lt^2+N)$. The MSSR algorithm, reviewed here for completeness, is described in Algorithm~\ref{alg:mssr}.

 
\begin{algorithm}
\caption{Collaborative IRS Decoder (Schmidt \emph{et. al}~\cite{schmidt2009collaborative})}
\textbf{Input}:
${S}^{(l)}=\{S^{(l)}_i\}_{i=0}^{N-K-1} \ \forall  l\in [L]$ \\
\textbf{Output}: $\mathbf{\hat{D}} \in \{\mathbb{F}^{L \times N},F \text{ (decoding failure)}\}$
\begin{algorithmic}[1]
\State {Synthesize $t$ and $\hat{\Lambda}(z)$ using the shift register synthesis algorithm in \cite{schmidt2006}}
\State {[$t$, $\hat{\Lambda}(z)$]\ = Shift Register Synthesis Algorithm(${S}^{(1)}, \dots,{S}^{(L)}$)} 
\State {$\hat{\mathbf{D}} = F$} 
\If{$t\leq t_{\max}$ and $\hat{\Lambda}(z)$ is $t$-valid}
\For{each $l$ from $1$ to $L$}
\State {From $\hat{\Lambda}(z)$ compute $\hat{\mathbf{E}}{(l,:)}$}
\State {Compute $\hat{\mathbf{D}}(l,:) = \hat{\mathbf{R}}(l,:) - \hat{\mathbf{E}}{(l,:)}$}
\EndFor
\EndIf
\end{algorithmic}
\label{alg:mssr}
\end{algorithm}

It can be seen that in the absence of numerical round-off errors, the outputs of the collaborative Peterson's algorithm and the MSSR algorithm are identical for every $\mathbf{R}$ since both of them compute the solution to (\ref{eq:4}). 

\section{Analysis of probability of failure and error for finite fields ($\mathbb{F} = \mathbb{F}_q$)}\label{pe_pf_section}
\indent In Section~\ref{poly ir}, we showed that Polynomial codes are IRS codes. Hence the fault tolerance of the Polynomial codes can be analyzed using similar techniques for IRS codes. In this section, we consider the uniformly random error model for finite fields (UREF), defined in Section~\ref{fte}, which was originally considered in~\cite{schmidt2009collaborative}. 
In particular, we define the error events 
\begin{align}
\label{eqn:errorsets}
\begin{split}
    &\mathcal{E}_1(t)  = \{\mathbf{E}: W_H(\mathbf{E})=t \ \text{and the MSSR/collaborative } 
          \text{algorithm fails}  \}, \\
    &\mathcal{E}_2(t) = \{\mathbf{E}: W_H(\mathbf{E})=t \ \text{and the MSSR/collaborative } 
        \text{algorithm makes an undetected error} \}, \\
    & \mathcal{E}(t) = \{\mathbf{E}: W_H(\mathbf{E})=t \}.
\end{split}    
\end{align}

Since the outputs of the collaborative Peterson's algorithm and the MSSR algorithm are identical for every $\mathbf{R}$, both algorithms have the same probability of failure and the same probability of undetected error. We denote by $P_{F}(t)$ and $P_{\rm ML}(t)$ the probability of failure and the probability of undetected error, respectively, given that $W_H(\mathbf{E})=t$. Under the UREF model, $P_{F}(t)$ and $P_{\rm ML}(t)$ are given by \cite{schmidt2009collaborative}
\[P_{F}(t)=\frac{|\mathcal{E}_1(t)|}{|\mathcal{E}(t)|}, \ P_{\rm ML}(t)=\frac{|\mathcal{E}_2(t)|}{|\mathcal{E}(t)|}.
\]

\subsection{Probability of Failure}
\label{sec:finitefieldpf}

A necessary condition for the failure of both the collaborative Peterson's algorithm and the MSSR algorithm is that the matrix $\mathbf{S}_L(t)$ is not full rank, as shown in \cite{schmidt2009collaborative}. To calculate an upper bound on $P_F(t)$, we refer to the analysis by schmidt {\em et al.} in~\cite{schmidt2009collaborative}, and recall the following result from \cite{schmidt2009collaborative}.

\begin{theorem}\cite[Theorem 7]{schmidt2009collaborative}
\label{eq:pf}
Under the UREF model, for all $t \leq t_{\max}= \frac{L}{L+1}(N-K)$, 
\begin{equation}
P_F(t) \leq \Bigg( \frac{q^L - \frac{1}{q}}{q^L-1}\Bigg) \frac{q^{-(L+1)(t_{\max}-t)}}{q-1}. 
\end{equation}
\end{theorem}
By the result of Theorem~\ref{eq:pf}, it can be readily seen that for all $t < t_{\rm max}$, $P_F(t)$ diminishes as $q^{-\Omega\big(L\big)}$ and for $t = t_{\rm max}$, $P_F(t)$ decays as $q^{-1}$ .

\subsection{Probability of Undetected Error}
\label{sec:finitefieldpe}
As shown in \cite[ Theorem~5]{schmidt2009collaborative}, the MSSR algorithm has the Maximum Likelihood (ML) certificate property, i.e., whenever the decoder of \cite{schmidt2006} does not fail, it yields the ML solution, namely the codeword at minimum Hamming distance from the received word. The collaborative Peterson's algorithm has the same ML certificate property as well.
An error matrix $\mathbf{E}$ with $W_H(\mathbf{E})=t$ is said to be a \emph{bad error matrix of Hamming weight $t$} if there exists a non-zero codeword $\mathbf{D} \in \mathcal{C}_{\rm IRS}$ such that $W_H(\mathbf{D} \ominus \mathbf{E}) \leq t$. 

We now use a result from \cite[Page 141]{vg2018} without proof. 
\begin{lemma}{\cite[Page 141]{vg2018}}
\label{lemma:Bs}
Let $\mathcal{C} \subseteq \{0,1,\cdots q-1\}^N$ be a code with relative distance $\delta = d_{\min}/N$, and let $S \subseteq [N]$ be such that $|S|=(1-\gamma)N$, where $0<\gamma \leq \delta-\epsilon$ for some $\epsilon>0$. Let $\mathcal{E}_S$ be the set of all error vectors with support $S^c$, and let $\mathcal{B}_S$ be the set of all bad error vectors with support $S^c$. Then,
\[
|\mathcal{B}_S| \leq q^{\frac{N}{\log_{2} q}- \frac{\epsilon N}{2}+ \frac{1}{2} }|\mathcal{E}_S|. 
\]
\end{lemma}

\begin{theorem}
Under the UREF model, for all $t\leq N-K-1$ (and in particular, for all $t\leq t_{\max}=\frac{L}{L+1}(N-K)$), $P_{\rm ML}(t) \to 0$ as $q^L \to \infty$.
\end{theorem}

\begin{proof}

It is easy to see that an IRS code can be viewed as a single code over $\mathbb{F}_{q^L}$, i.e. $\mathcal{C}_{\rm IRS}$ is a {$\big(\mathbb{F}_{q^L},N,K,N-K+1\big)$ code.} Lemma~\ref{lemma:Bs} holds for a single code and, hence, can be applied to $\mathcal{C}_{\rm IRS}$ with $q$ being replaced by $q^L$. Since the upper bound in Lemma~\ref{lemma:Bs} depends only on the cardinality of $\mathcal{E}_S$, it follows that 
the probability of having a bad error matrix with $W_H(\mathbf{E})=t$ for the $\big(\mathbb{F}_{q^L},N,K,N-K+1\big)$ code (replacing $q$ by $q^L$ since $\mathcal{C}_{\rm IRS}$ is over $q^L$) which we denote by $P_e(t)$ is upper bounded by
\begin{equation}\label{eq:Pet}
P_e(t) = \frac{|\mathcal{B}_S|}{|\mathcal{E}_S|}\leq q^{
 L(\frac{N}{\log_{2} q^L}- \frac{\epsilon N}{2}+ \frac{1}{2})}.
\end{equation}
By setting $\delta = \frac{N-K+1}{N}$ and $\epsilon = \frac{2}{N}$,
it is easy to see that $P_e(t) \rightarrow 0$ as $q^L \to \infty$. For this choice of $\delta$ and $\epsilon$, it follows that $\gamma \leq \delta-\epsilon=\frac{N-K-1}{N}$, which implies that~\eqref{eq:Pet} holds for all $t \leq N-K-1$.


Note that the algorithms in Section~\ref{sec_alg} have the ML certificate property. Note, also, that the fraction of error matrices that give rise to an undetected error is upper bounded by the fraction of bad error matrices. This is simply because without a bad error matrix of Hamming weight up to $(\delta-\epsilon)N$, an undetected error cannot occur. Thus, $P_{\rm ML}(t)\leq P_e(t)$. Since $P_e(t)$ vanishes as $q^L \to \infty$, then $P_{\rm ML}(t)$ vanishes as $q^L \to \infty$. Moreover, $N$ and $K$ are fixed and finite, and hence, $\sum_{t=1}^{N-K-1} P_{\rm ML}(t) \to 0 $ as $ q^L \to \infty$.
\end{proof}

\section{Analysis of probablity of failure and probability of error for the real field}
In this section, we analyze the probability of failure and probability of error under the GRE model when the computations are performed over the real field. In particular, we consider the case that the error values are independently and identically distributed standard Gaussian random variables (with zero mean and unit variance). Note, however, that this assumption does not limit the generality of the results, and is made for the ease of exposition only. For this model, conditioned on $t$ errors occurring, the probability of failure ($P_{F}(t)$) and the probability of undetected error ($P_{\mathrm{ML}}(t)$) are given by
\begin{equation*}
P_{F}(t)= \frac{\int_{\mathcal{E}_1(t)} \phi(\mathbf{x}) \ d\mathbf{x}}{\int_{\mathcal{E}(t)} \phi(\mathbf{x}) \ d\mathbf{x} }, \quad P_{\rm ML}(t) = \frac{\int_{\mathcal{E}_2(t)} \phi(\mathbf{x}) \ d\mathbf{x}}{\int_{\mathcal{E}(t)} \phi(\mathbf{x}) \ d\mathbf{x}},
\end{equation*} where $\mathcal{E}_1(t), \mathcal{E}_2(t), \mathcal{E}(t)$ are defined as in (\ref{eqn:errorsets}), and $\phi(\mathbf{x})$ is the probability density function of an $Lt$-dimensional standard Gaussian random vector (with zero-mean vector and identity covariance matrix).

\subsection{Probability of Failure}
It should be noted that the results of~\cite{schmidt2009collaborative} for finite fields cannot be directly extended to the real field, simply because the counting arguments used in~\cite{schmidt2009collaborative} for finite fields do not carry over to the real field. In this section, we propose a new approach to derive the probability of failure for the real field case.


For simplifying the notation, hereafter, we use $\rho\triangleq N-K-t$. Suppose that $t \leq t_{\rm max}=\frac{L}{L+1}(N-K)$ errors occur at positions $j_1,j_2,\cdots,j_t$ with values $e_{j_1}^{(l)},e_{j_2}^{(l)},\cdots,e_{j_t}^{(l)}$ for the $l$th RS code. 
Recall the syndrome matrix $\mathbf{S}^{(l)}(t)$ for the $l$th RS code (see~\eqref{eq:SlTl}).  
As shown in~\cite{schmidt2009collaborative}, $\mathbf{S}^{(l)}(t)$ can be decomposed as \[\mathbf{S}^{(l)}(t)=\mathbf{H}^{(l)}(t) \cdot \mathbf{F}^{(l)}(t) \cdot \mathbf{D}(t) \cdot \mathbf{Y}(t),\] where $\mathbf{H}^{(l)}(t) \triangleq (\alpha^{(i-1)}_{j_k})_{i\in [\rho], k\in [t]}$ is an $\rho\times t$ matrix, $\mathbf{F}^{(l)}(t)\triangleq\text{diag}((e_{j_i}^{(l)})_{i\in [t]})$ is a $t\times t$ diagonal matrix, $\mathbf{D}(t)\triangleq\text{diag}((\alpha_{j_i})_{i\in [t]})$ is a $t\times t$ diagonal matrix, and $\mathbf{Y}(t)\triangleq(\alpha^{(k-1)}_{j_i})_{i\in [t],k\in [t]}$ is a $t\times t$ matrix. 

\begin{theorem}
Under the GRE model, for all $t\leq t_{\max} = \frac{L}{L+1}(N-K)$, $P_{F}(t) = 0$. In particular, for $L\geq N-K-1$, for all $t\leq N-K-1$, $P_{F}(t) = 0$. 
\end{theorem}
\begin{proof}
The decoding algorithms described in Section~\ref{sec_alg} fail when the stacked matrix $\mathbf{S}_L(t)$ defined in (\ref{eq:4}) is rank deficient, i.e.,  there exists a non-zero row-vector $\mathbf{v}$ such that $\mathbf{S}_L(t)\cdot\mathbf{v}^{\mathsf{T}}=0$. Alternatively, $\mathbf{S}_L(t)$ is rank deficient iff there exists a non-zero row-vector $\mathbf{v}$ such that
\begin{equation}\label{eq:1}
  \mathbf{S}^{(l)}(t)\cdot \mathbf\mathbf{v}^{\mathsf{T}}=(\mathbf{H}^{(l)}(t) \cdot \mathbf{F}^{(l)}(t) \cdot \mathbf{D}(t) \cdot \mathbf{Y}(t)) \cdot \mathbf{v}^{\mathsf{T}} =0 \quad \forall  l\in [L].
\end{equation} 
Since $\mathbf{D}(t)$ and $\mathbf{Y}(t)$ are invertible, the condition~\eqref{eq:1} holds iff there is a non-zero row-vector $\mathbf{v}$ such that 
\begin{equation} \label{eq:2}
 (\mathbf{H}^{(l)}(t) \cdot \mathbf{F}^{(l)}(t))\cdot \mathbf{v}^{\mathsf{T}}=0 \quad \forall l\in [L].
\end{equation}
Let $\mathbf{v}=(v_1,v_2,\dots,v_t)$, and let $f_{i,l}\triangleq e^{(l)}_{j_i}$ for all $i\in [t]$. 
Expanding~\eqref{eq:2}, it is easy to see that 
\begin{equation} \label{eq:8}
\underbrace{
\begin{pmatrix}
v_1&v_2&\cdots&v_t \\
v_1 \cdot \alpha_{j_1} &v_2 \cdot \alpha_{j_2} & \cdots & v_t \cdot \alpha_{j_t} \\
v_1 \cdot \alpha^{2}_{j_1} &v_2 \cdot  \alpha^{2}_{j_2} & \cdots &v_t \cdot  \alpha^{2}_{j_t} \\
\vdots & \vdots & & \vdots \\
v_1 \cdot \alpha^{(\rho-1)}_{j_1} &v_2 \cdot  \alpha^{(\rho-1)}_{j_2} & \cdots & v_t \cdot \alpha^{(\rho-1)}_{j_t}
\end{pmatrix}}_{\mathbf{H}} \underbrace{
\begin{pmatrix}
f_{1,l} \\
f_{2,l}  \\
\vdots\\
f_{t,l}
\end{pmatrix}}_{\mathbf{f}^{(l)}} =0. 
\end{equation} Combining the condition~\eqref{eq:8} for all the RS codes in the IRS code (for all $l\in [L]$), it holds that
\begin{equation} \label{eq:9}
\mathbf{H}\cdot \mathbf{F} = 0,    
\end{equation} where $\mathbf{H}$ is defined in~\eqref{eq:8}, and $\mathbf{F}\triangleq (\mathbf{f}^{(1)},\mathbf{f}^{(2)},\dots,\mathbf{f}^{(L)})$ is a $t\times L$ matrix where $\mathbf{f}^{(l)}$ for $l\in [L]$ is defined in~\eqref{eq:8}. 
Alternatively,~\eqref{eq:9} can be written as
\begin{equation}\label{eq:10}
\mathbf{v}\cdot \mathbf{\Phi} = 0,    
\end{equation} where $\mathbf{\Phi}$ is a $t\times \rho L$ matrix given by
\begin{equation}\label{eq:Phi}
\mathbf{\Phi} \triangleq \begin{pmatrix}
f_{1,1} &\cdots &f_{1,L}& 
(\alpha_{j_1} f_{1,1} )&\cdots & (\alpha_{j_1} f_{1,L})&\cdots& 
(\alpha^{(\rho-1)}_{j_1} f_{1,1}) &\cdots &(\alpha^{(\rho-1)}_{j_1} f_{1,L})\\
f_{2,1} &\cdots &f_{2,L}& 
(\alpha_{j_2} f_{2,1}) &\cdots & (\alpha_{j_2} f_{2,L})&\cdots&
(\alpha^{(\rho-1)}_{j_2} f_{2,1}) &\cdots & (\alpha^{(\rho-1)}_{j_2} f_{2,L})\\
\vdots& & \vdots&\vdots& & \vdots& & \vdots& & \vdots\\
f_{t,1} &\cdots &f_{t,L}&
(\alpha_{j_t} f_{t,1} )&\cdots & (\alpha_{j_t} f_{t,L})&\cdots&
(\alpha^{(\rho-1)}_{j_t} f_{t,1}) &\cdots & (\alpha^{(\rho-1)}_{j_t} f_{t,L})
\end{pmatrix}.
\end{equation}

Let $\mathcal{F}$ be the set of all $t\times L$ matrices $\mathbf{F}=(f_{i,l})_{i\in [t],l\in [L]}$ for each of which the condition~\eqref{eq:9} holds for some non-zero vector $\mathbf{v}$. We need to show that $\mathcal{F}$ is a set of measure zero. 

We consider two cases as follows: (i) $t\leq L$, and (ii) $t>L$. 

\emph{Case (i):} For the condition~\eqref{eq:9} to hold, there must exist a non-zero vector $\mathbf{v}$ in the left null space of $\mathbf{F}$. It is easy to see that, under the GRE model, the set of all matrices $\mathbf{F}$ that have a row-rank of $t$ is a set of measure $1$. This implies that the set of all matrices $\mathbf{F}$ for each of which there exists some non-zero vector $\mathbf{v}$ in the left null space of $\mathbf{F}$ is a set of measure zero. Thus, for $t\leq L$, $\mathcal{F}$ is a set of measure zero.

\emph{Case (ii):} For a vector $\mathbf{v}$, let the weight of $\mathbf{v}$, denoted by $\mathrm{wt}(\mathbf{v})$, be the number of non-zero elements in $\mathbf{v}$. For any integer $1\leq w\leq t$, let $\mathcal{F}_{w}$ be the set of all matrices $\mathbf{F}$ for each of which there exists a non-zero vector $\mathbf{v}$ such that $\mathrm{wt}(\mathbf{v})=w$ and the condition~\eqref{eq:9} holds. 

We consider two cases as follows: (1) $w\leq \rho$, and (2) $w>\rho$. (Recall that $\rho = N-K-t$.)

\begin{enumerate}
\item[(1)] $w \leq \rho$: Assume, without loss of generality, that $v_1,v_2,\cdots,v_{w}$ are the non-zero elements of $\mathbf{v}$. 
Let $\mathbf{H}_{w}\triangleq ((v_k\cdot\alpha^{(i-1)}_{j_k})_{i\in [w],k\in [w]})$ be the $w\times w$ sub-matrix of $\mathbf{H}$ (defined in~\eqref{eq:9}) corresponding to the first $w$ rows and the first $w$ columns, and let $\mathbf{F}_w\triangleq ((f_{i,l})_{i\in [w],l\in [L]})$ be the $w\times L$ sub-matrix of $\mathbf{F}$ corresponding to the first $w$ rows. Then, the condition~\eqref{eq:9} reduces to
\begin{equation*}
\mathbf{H}_{w}\cdot \mathbf{F}_{w} = 0.
\end{equation*} 

It is easy to see that the matrix $\mathbf{H}_{w}$ generates a Generalized Reed-Solomon code with distinct parameters $\{\alpha_{j_i}\}_{i\in [w]}$ and non-zero multipliers $\{v_i\}_{i\in [w]}$. Thus, $\mathbf{H}_{w}$ is full rank (and hence, invertible). This implies that for each $l\in [L]$ the column-vector $\mathbf{f}^{(l)}$ (defined in~\eqref{eq:8}) is an all-zero vector. 
Thus, every matrix in $\mathcal{F}_{w}$ for $w\leq \rho$ contains a $w\times L$ all-zero sub-matrix. In particular, every matrix in $\mathcal{F}_{w}$ for $w\leq \rho$ has at least one fixed (zero, in this case) entry. Under the GRE model, it is then easy to see that $\mathcal{F}_{w}$ for $w\leq \rho$ is a set of measure zero.

\item[(2)] $w > \rho$: Assume, without loss of generality, that $v_1,\dots,v_w$ are the non-zero elements of $\mathbf{v}$, and let $\tilde{\mathbf{v}}\triangleq (v_1,v_2,\cdots,v_w)$. 
Let $\mathbf{\Phi}_w$ be the $w\times \rho L$ sub-matrix of $\mathbf{\Phi}$ (defined in~\eqref{eq:Phi}) corresponding to the first $w$ rows, 
\begin{align*}
\mathbf{\Phi}_w &\triangleq 
\begin{pmatrix}
f_{1,1} &\cdots &f_{1,L}& 
(\alpha_{j_1} f_{1,1} )&\cdots & (\alpha_{j_1} f_{1,L})&\cdots& 
(\alpha^{(\rho-1)}_{j_1} f_{1,1}) &\cdots &(\alpha^{(\rho-1)}_{j_1} f_{1,L})\\
f_{2,1} &\cdots &f_{2,L}& 
(\alpha_{j_2} f_{2,1}) &\cdots & (\alpha_{j_2} f_{2,L})&\cdots&
(\alpha^{(\rho-1)}_{j_2} f_{2,1}) &\cdots & (\alpha^{(\rho-1)}_{j_2} f_{2,L})\\
\vdots& & \vdots&\vdots& &\vdots& &\vdots& & \vdots\\
f_{w,1} &\cdots &f_{w,L}&
(\alpha_{j_w} f_{w,1} )&\cdots & (\alpha_{j_w} f_{w,L})&\cdots&
(\alpha^{(\rho-1)}_{j_w} f_{w,1}) &\cdots & (\alpha^{(\rho-1)}_{j_w} f_{w,L})
\end{pmatrix}.    
\end{align*} Then, the condition~\eqref{eq:10} reduces to 
\begin{equation}\label{eq:13}
\tilde{\mathbf{v}}\cdot \mathbf{\Phi}_w = 0.    
\end{equation}

Since in~\eqref{eq:4} the number of variables must be less than the number of equations, then $w \leq t \leq \rho L$. 
Note that $\mathbf{\Phi}_w$ is a $w \times \rho L$ matrix. Thus, $\mathrm{rank}(\mathbf{\Phi}_w)\leq w$. Moreover, there exists a non-zero vector $\tilde{\mathbf{v}}$ in the left null space of $\mathbf{\Phi}_w$. This implies that $\mathrm{rank}(\mathbf{\Phi}_w)\leq w-1$. Since the row-rank and the column-rank are equal, there exists a non-zero column-vector  $\mathbf{u}$ such that \[\mathbf{\Phi}_w\cdot \mathbf{u} =0.\]


Let $\alpha_i \triangleq \alpha_{j_i}$ for $i\in [w]$, and let $\boldsymbol{\alpha}^{(k)}=(\alpha_1^{k-1},\alpha_2^{k-1},\cdots,\alpha_w^{k-1})^{\mathsf{T}}$ for $k\in [\rho]$. We define the product operator $\odot$ between the two vectors $\boldsymbol{\alpha}^{(k)}$ and $\mathbf{f}^{(l)}$ as
\[
\boldsymbol{\alpha}^{(k)}\odot \mathbf{f}^{(l)} \triangleq (\alpha_1^{(k-1)}f_{1,l},\alpha_2^{(k-1)}f_{2,l},\dots,\alpha_w^{(k-1)}f_{w,l})^{\mathsf{T}}.
\]
Then, we can rewrite $\mathbf{\Phi}_w$ as
\begin{equation*}\label{eq:14}
\left(\boldsymbol{\alpha}^{(1)}\odot \mathbf{f}^{(1)},\dots,\boldsymbol{\alpha}^{(1)}\odot \mathbf{f}^{(L)},\boldsymbol{\alpha}^{(2)}\odot \mathbf{f}^{(1)},\dots,\boldsymbol{\alpha}^{(2)}\odot \mathbf{f}^{(L)},\dots,\boldsymbol{\alpha}^{(\rho)}\odot \mathbf{f}^{(1)}\dots,\boldsymbol{\alpha}^{(\rho)}\odot \mathbf{f}^{(L)}\right).
\end{equation*}

Since $\mathbf{u} = (u_{1},\dots,u_L,u_{L+1},\dots,u_{L+L},\dots,u_{(\rho-1)L+1},\dots,u_{(\rho-1)L+L})\neq 0$, there exist $l\in [L]$ and $k\in [\rho]$ such that $u_{(k-1)L+l}$ is non-zero. Assume, without loss of generality, that $u_{1}\neq 0$. 
Consider the columns $\boldsymbol{\alpha}^{(1)} \odot \mathbf{f}^{(1)},\boldsymbol{\alpha}^{(2)} \odot \mathbf{f}^{(1)},\dots,\boldsymbol{\alpha}^{(\rho)} \odot \mathbf{f}^{(1)}$ in the matrix $\mathbf{\Phi}_w$, and their corresponding elements $u_{1},u_{L+1},\dots,u_{(\rho-1)L+1}$ in the vector $\mathbf{u}$. Let $\tilde{u}_k\triangleq u_{(k-1)L+1}$ for $k\in [\rho]$, and let $\tilde{\mathbf{u}}\triangleq (\tilde{u}_1,\dots,\tilde{u}_{\rho})$. Note that $\tilde{\mathbf{u}}\neq 0$ (by construction). Consider the vector \[\mathbf{g}\triangleq\tilde{u}_1 (\boldsymbol{\alpha}^{(1)} \odot \mathbf{f}^{(1)}) + \tilde{u}_2 (\boldsymbol{\alpha}^{(2)} \odot \mathbf{f}^{(1)}) +\cdots +\tilde{u}_{\rho} (\boldsymbol{\alpha}^{(\rho)} \odot \mathbf{f}^{(1)}).\] 
Expanding $\mathbf{g} = (g_1,\dots,g_w)^{\mathsf{T}}$, we get $g_i=(\tilde{u}_1\alpha_i^0+\tilde{u}_2\alpha_i^1+\dots+\tilde{u}_{\rho}\alpha_i^{\rho-1})f_{i,1}$ for all $i\in [w]$.
Note that there exists $i\in [w]$ such that the coefficient of $f_{i,1}$ in $g_i$, i.e., $\tilde{u}_1\alpha_i^0+\tilde{u}_2\alpha_i^1+\dots+\tilde{u}_{\rho}\alpha_i^{\rho-1}$, is non-zero. The proof is by the way of contradiction. Suppose that for all $i\in [w]$ the coefficient of $f_{i,1}$ in $g_i$ is zero. Let $\mathbf{M}\triangleq ((\alpha_i^{k-1})_{i\in [w],k\in [\rho]})$. Then it is easy to see that $\mathbf{M}\cdot\tilde{\mathbf{u}} = 0$. 
Since $\mathbf{M}$ is a $w\times \rho$ Vandermonde matrix with $\rho < w$, then $\mathrm{rank}(\mathbf{M})=\rho$. This implies that $\tilde{\mathbf{u}}=0$. This is however a contradiction because $\tilde{\mathbf{u}}\neq 0$ (by assumption). Thus, for some $i\in [w]$ the coefficient of $f_{i,1}$ in $g_i$ must be non-zero. Thus, every matrix in $\mathcal{F}_w$ for $w>\rho$ contains at least one entry which can be written as a linear combination of the rest of the entries. Under the GRE model, this readily implies that $\mathcal{F}_w$ is a set of measure zero. 
\end{enumerate}

Noting that $\mathcal{F}=\cup_{w=1}^{t} \mathcal{F}_w$ and taking a union bound over all $w$ ($1\leq w \leq t$), it follows that for $t>L$, $\mathcal{F}$ is a set of measure zero. This completes the proof. 
\end{proof}


\subsection{Probability of Undetected Error}
Similarly as in the case of the finite fields, both the MSSR decoding algorithm and the collaborative Peterson's decoding algorithm give an error locator polynomial $\Lambda(z)$ over the real field ($\mathbb{R}$) of the least possible degree which satisfies all the syndrome equations in~\eqref{eq:4}. This implies that these decoding algorithms have the ML certificate property (for details, see Section~\ref{sec:finitefieldpe}).

As was shown by Dutta {\em et al.} in~\cite[Theorem~3]{dutta2018unified}, under the GRE model, when the number of errors (i.e., the Hamming weight of the error matrix) is less than $N-K$, with probability $1$ the closest codeword to the received vector is the transmitted codeword. This implies that for any decoding algorithm satisfying the ML certificate property, the set of all bad error matrices (defined in Section~\ref{sec:finitefieldpe}) is of measure zero, and thereby, the probability of undetected error is zero.

\begin{theorem}
Under the GRE model, for all $t \leq N-K-1$ (and in particular, for all $t\leq t_{\max}=\frac{L}{L+1}(N-K)$), $P_{\mathrm{ML}}(t) = 0$. 
\end{theorem}

\section{Numerical Results}
\begin{figure}
\begin{center}
\begin{tikzpicture}[scale=0.7, transform shape]
\begin{axis}[
scale only axis,
xmin=0, xmax=6, ymin=0, ymax=1, ymode=linear,yminorticks=true,
xlabel=Number of errors, ylabel=Probability of error - $P_F(t) + P_{\rm ML}(t)$, legend pos=north west]
\addplot[color=blue,dotted,mark=*,mark options={solid},line width = 2pt]
table[row sep=crcr]{%
    1.0000    0\\
    2.0000    0\\
    3.0000    0\\
    4.0000    1\\
    5.0000    1\\
    6.0000    1\\
};
\addlegendentry{Individual GRS decoding ($L=1$)}
\addplot [color=red,solid,mark=triangle*,mark size = 3pt, mark options={solid},line width = 2pt]
table[row sep=crcr]{%
    1.0000         0\\
    2.0000    0\\
    3.0000    0\\
    4.0000    0\\
    5.0000    0\\
    6.0000    1\\
};
\addlegendentry{CPDA $L=6$}
\addplot [color=black,dashed,mark=square*,mark options={solid},line width = 2pt]
table[row sep=crcr]{%
    1.0000         0\\
    2.0000    0.06\\
    3.0000    0.23\\
    4.0000    0.47\\
    5.0000    0.71\\
    6.0000    0.9772\\
};
\addlegendentry{$\ell_1$minimization decoder of~\cite{candes2005decoding}}
\end{axis}
\end{tikzpicture}
\end{center}
\caption{Probability of error for CPDA and $\ell_1$-minimization decoders, $N=8\ K=2$} 
\label{fig:combined82}
\end{figure}
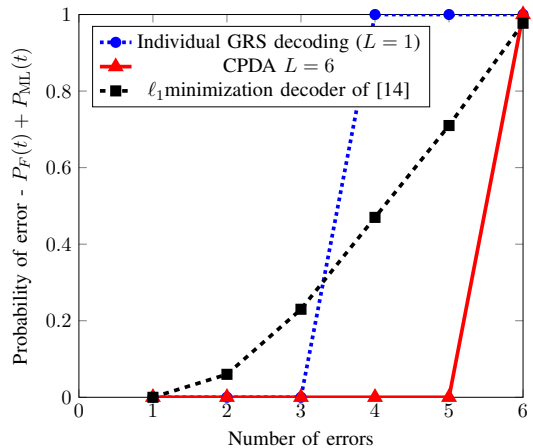
We present simulation results for $N=8$, $K=2$, and $\alpha_i=0.9^i$ for different $L$.  Fig.~\ref{fig:combined82} shows the probability of error ($P_e(t) = P_F(t) + P_{\rm ML}(t)$) for decoding RS codes individually using Peterson's algorithm ($L=1$), decoding RS codes individually using the $\ell_1$ minimization decoder, and collaborative decoding using the CPDA algorithm with $L=6$. For each data point, 12500 IRS codewords were simulated.  It can be seen that the CPDA with $L=6$ corrects all $t$ errors for $t \leq N-K-1$, which is a significant improvement over decoding RS codes individually. This is consistent with the theoretical results. The probability of error for the $\ell_1$ minimization decoder remains fairly high for several values of $t \leq N-K-1$. These results are consistent with the results of Candes and Tao (Figures 2 and 3 in~\cite{candes2005decoding}). This shows that individually decoding RS decoder 
using the $\ell_1$-minimization decoder does not suffice to achieve small probability of error as suggested in~\cite{dutta2018unified}; whereas, collaborative decoding can achieve the decoding radius bound of $N-K-1$ with polynomial complexity.

For larger values of $N$ and $K$, we noticed that computing the rank of $\mathbf{S}_L(t)$ had numerical inaccuracies. This is a well-known issue with decoding RS codes over the real field. Interestingly, from simulations, we observe that collaborative decoding seems to alleviate this issue. Table~\ref{table:CPDA} shows the probability of error ($P_e(t)=P_F(t)+P_{\rm ML}(t)$) for $N=20$, $K=12$ and $\alpha_i=i$. For a fixed number of errors, increasing $L$ improved the condition number of $\mathbf{S}_L(t)^{\sf T} \mathbf{S}_L(t)$. With $L=20$, we were able to decode up to $N-K-1$ errors with $P_e(t)=0$ in 12500 trials. 

\begin{table}[h]
    \centering
    \caption{Probability of error for the CPDA, $N=20\ K=12 $, 12500 trials}
    \begin{tabular}{|c|c|c|c|c|c|c|c|}
    \hline
     $L \backslash t$     & 1 & 2 & 3 & 4 & 5 & 6 &  7\\
    \hline
    1 & 0 & 0 & 0 & 0.0008  & - & - & - \\
    \hline
    2 & 0 & 0 & 0 & 0 & 0 & - & - \\
    \hline
    3 & 0 & 0 & 0 & 0 & 0 & 0 & - \\
    \hline
    4 & 0 & 0 & 0 & 0 & 0 & 0 & - \\
    \hline
    5 & 0 & 0 & 0 & 0 & 0 & 0 & - \\
    \hline
    6 & 0 & 0 & 0 & 0 & 0 & 0 & - \\
    \hline
    7 & 0 & 0& 0 & 0 & 0 & 0 & 0.0026 \\
    \hline
    8 & 0 & 0 & 0 & 0 & 0 & 0 & 0.0008 \\
    \hline
    20 & 0 & 0 & 0 & 0 & 0 & 0 & 0 \\
    \hline
    \end{tabular}
    \label{table:CPDA}
\end{table}

 \begin{figure}
\begin{center}
\begin{tikzpicture}
\begin{axis}[
scale only axis,
xmin=0, xmax=6, ymin=1e-1, ymax=1e+15, ymode=log,yminorticks=true,
xlabel=Number of errors, ylabel= Average condition number of $\mathbf{S}_L^T(t)\mathbf{S}_L(t)$,legend pos=south east]
\addplot[color=red,solid,mark=*,mark options={solid}]
table[row sep=crcr]{%
1 0.000000000000010e+14\\
2 0.000007544282585e+14\\
3 0.406160716134754e+14\\
};
\addlegendentry{$L=1$}
\addplot [color=black,solid,mark=*,mark options={solid}]
table[row sep=crcr]{%
1 0.000000000000010e+14\\
2 0.000000000316837e+14\\
3 0.000001530774385e+14\\
4 0.347521761551090e+14\\
};
\addlegendentry{$L=2$}
\addplot [color=blue,solid,mark=*,mark options={solid}]
table[row sep=crcr]{%
1 0.000000000000010e+14\\   
2 0.000000000043387e+14\\   
3 0.000000077336598e+14\\ 
4 0.000125288513636e+14\\
};
\addlegendentry{$L=3$}
\addplot [color=teal,solid,mark=*,mark options={solid}]
table[row sep=crcr]{%
1 0.000000000000010e+14\\
2 0.000000000020670e+14\\
3 0.000000017455253e+14\\
4 0.000013714872979e+14\\  
};
\addlegendentry{$L=4$}
\addplot [color=brown,solid,mark=*,mark options={solid}]
table[row sep=crcr]{%
1 0.000000000000010e+14\\
2 0.000000000014349e+14\\
3 0.000000008465952e+14\\
4 0.000004139746603e+14\\
5 1.877525707810815e+14\\
};
\addlegendentry{$L=5$}
\end{axis}
\end{tikzpicture}
\end{center}
\caption{Average condition number of $\mathbf{S}_L^T(t)\mathbf{S}_L(t)$, $N=8\  K=2$ } \label{fig:cnum}
\end{figure}
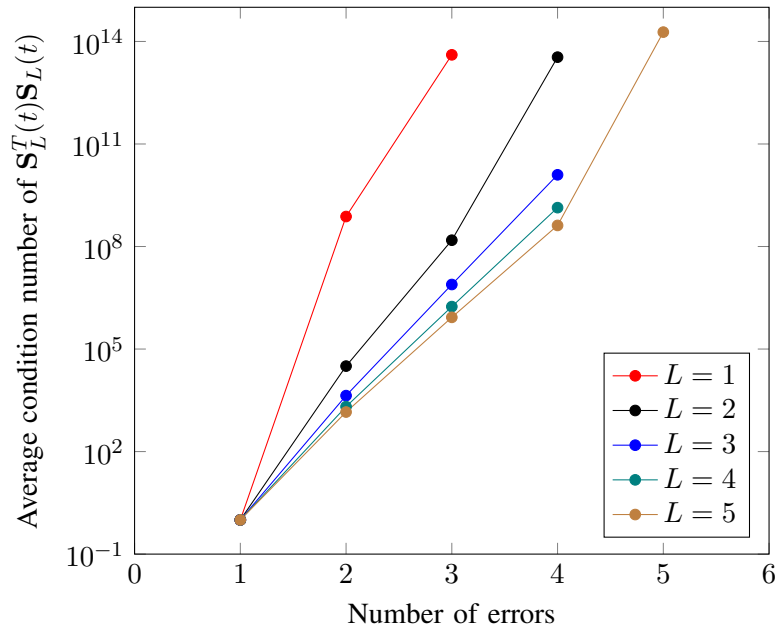
 
 Our results have shown that collaborative decoding of Polynomial codes can correct up to $t_{\max}=\frac{L}{L+1}(N-K)$ errors. It can be seen that $t_{\rm max} = N-K-1$ for all $L \geq N-K-1$ and hence, it is natural to wonder if there is any advantage in increasing $L$ beyond $N-K-1$. Here we empirically show that increasing $L$ improves the numerical stability of the collaborative Peterson's algorithm for determining the error locator polynomial. Fig.~\ref{fig:cnum} ($N=8$, $K=2$, $\alpha_i=0.9^i$) shows a plot of the average condition number of the stacked syndome matrix $\mathbf{S}_L(t)$ (defined in~\eqref{eq:4}) as a function of $t$ for different $L$. It can be seen from simulations that for all $t$, increasing $L$ decreases the average condition number. Since the collaborative Peterson's algorithm requires inversion of the matrix $\mathbf{S}_{L}^{\mathsf{T}}(t) \mathbf{S}_L(t)$, the numerical stability of the algorithm will improve with increasing $L$.
\bibliographystyle{IEEEtran}
\bibliography{IEEEabrv,collab}

\end{document}